\documentclass{article}

\usepackage{fullpage}
\usepackage{amsthm}
\usepackage{amsfonts}
\usepackage[cmex10]{amsmath}
\usepackage{amssymb}
\usepackage{srcltx}

\newcommand{\avgd}{\bar{d}}

\newcommand{\Elr}{{\vec{E}}_{LR}}
\newcommand{\Erl}{{\vec{E}}_{RL}}
\newcommand{\de}{\vec{e}}
\newcommand{\dG}{\vec{G}}

\newcommand{\expt}{\mathbb{E}}
\newcommand{\E}{\expt}
\newcommand{\comment}[1]{}
\theoremstyle{plain}
\newtheorem{lem}{Lemma}[section]
\newtheorem{theorem}[lem]{Theorem}
\newtheorem{lemma}[lem]{Lemma}

\newtheorem{observation}[lem]{Observation}

\theoremstyle{remark}
\newtheorem{remark}[lem]{Remark}

\title{An entropy based proof of the Moore bound for irregular graphs}

\author{S. Ajesh Babu\footnote{This work was done while the author was at the School of Technology and Computer Science, 
Tata Institute of Fundamental Research, Homi Bhabha Road, Mumbai 400005, India.} 
\\\texttt{ajesh@yahoo-inc.com}\\Yahoo! Labs,\\Bangalore
\and
Jaikumar Radhakrishnan\footnote{School of Technology and Computer Science, 
Tata Institute of Fundamental Research, Homi Bhabha Road, Mumbai 400005, India.} 
\\\texttt{jaikumar@tifr.res.in}\\Tata Institute of Fundamental Research,\\Mumbai}

\date{\today}

\begin{document}
\maketitle
\begin{abstract}
We provide proofs of the following theorems % of Alon, Hoory and Linial 
by considering the entropy of random walks.\\
\begin{description}
 \item \textbf{Theorem 1.(Alon, Hoory and Linial)}
Let $G$ be an undirected simple graph with $n$ vertices, girth $g$, minimum degree at least $2$ and
average degree $\avgd$.
\begin{description}
\item[\mbox{\em Odd girth:}] If $g=2r+1$, then $n\geq 1+\avgd\displaystyle\sum_{i=0}^{r-1}(\avgd-1)^{i}$.
\item[\mbox{\em Even girth:}] If $g=2r$, then $n\geq 2\displaystyle\sum_{i=0}^{r-1} (\avgd-1)^i$.
\end{description}
 \item \textbf{Theorem 2.(Hoory)} Let $G=(V_L,V_R,E)$ be a bipartite
graph of girth $g=2r$, with $n_L=|V_L|$ and $n_R=|V_R|$, minimum
degree at least $2$ and the left and right average degrees $d_L$ and
$d_R$. Then,
\begin{eqnarray*}
n_L&\geq& \sum_{i=0}^{r-1}(d_R-1)^{\lceil \frac{i}{2}\rceil}(d_L-1)^{\lfloor \frac{i}{2}\rfloor},\\
n_R&\geq& \sum_{i=0}^{r-1}(d_L-1)^{\lceil \frac{i}{2}\rceil}(d_R-1)^{\lfloor \frac{i}{2}\rfloor}.
\end{eqnarray*}
\end{description}
\end{abstract}

\section{Introduction}
The Moore bound (see Theorem \ref{thm:moore}) gives a lower bound on the order of any simple undirected graph, based on its 
minimum degree and girth. Alon, Hoory and Linial~\cite{AHL02} showed that the same bound holds with the minimum degree replaced by 
the average degree. Later, Hoory~\cite{Hoory02} obtained a better bound for simple bipartite graphs. 
We reprove the results of Alon, Hoory and Linial~\cite{AHL02} and Hoory~\cite{Hoory02} using information theoretic arguments 
based on non-returning random walks on the graph.

The paper has three sections: In Section \ref{sec:notation} we introduce the relevant notation and 
terminology. In Section \ref{sec:moore_irregular}, we present the information theoretic 
proof of the result of Alon, Hoory and Linial~\cite{AHL02}; in Section \ref{sec:moore_bipartite}, we present a similar proof of 
the result of Hoory~\cite{Hoory02} for bipartite graphs.

\section{Notation} 
\label{sec:notation}
For an undirected simple graph $G=(V,E)$, let $\dG=(V,\vec{E}),$ be the directed version of $G$, where for each 
undirected edge of the form $\{u,v\}$ in $E$, we place two directed edges in $\vec{E}$, one of the form $(u,v)$ and another of the form 
$(v,u)$. Similarly, for an undirected bipartite graph $G=(V_L,V_R,E)$, let $\dG=(V_L,V_R,\Elr\cup \Erl)$ be the directed version 
of $G$, where for each undirected edge of the form $\{u,v\}$ in $E$, with $u\in V_L$ and $v\in V_R$, we place one directed edge of 
the form $(u,v)$ in $\Elr$, and another of the form $(v,u)$ in $\Erl$.

We will consider {\em non-returning} walks on $\dG$, that is, walks where the edges corresponding to the same undirected edge of 
$G$ do not appear in succession. For a vertex $v$, let $n_i(v)$ denote the number of non-returning walks in $\dG$ starting at $v$ 
and consisting of $i$ edges. For an edge $\de$, let $n_i(\de)$ denote the number of non-returning walks in $\dG$ starting with 
$\de$ and consisting of exactly $i+1$ edges (including $\de$).

\section{Moore bound for irregular graphs}
\label{sec:moore_irregular}
In Section~\ref{subsec:moore}, we recall the proof of the Moore bound;
in Section~\ref{subsec:ahl}, we review and reprove the theorem of
Alon, Hoory and Linial~\cite{AHL02} assuming the Lemma \ref{lm:ahl}.
In Section \ref{subsec:entropy}, we prove this lemma using an entropy
based argument.
\subsection{Proof of the Moore bound}
\label{subsec:moore}
The Moore bound provides a lower bound for the order of a graph in terms of its minimum degree and girth.
\begin{theorem}[{The Moore bound~\cite[p.~180]{B93}}]\label{thm:moore}Let $G$ be a simple undirected graph with $n$ vertices, 
minimum degree $\delta$ and girth $g$.
\begin{description}
\item[\mbox{\em Odd girth:}] If $g=2r+1$, then $\displaystyle n \geq 1+ \delta\sum_{i=0}^{r-1}(\delta-1)^{i}$.
\item[\mbox{\em Even girth:}] If $g=2r$, then
$\displaystyle  n \geq 2\sum_{i=0}^{r-1} (\delta-1)^i.$
\end{description}
\end{theorem}
The key observation in the proof of the Moore bound is the following.
If the girth is $2r+1$, then two distinct non-returning walks of
length at most $r$ starting at a vertex $v$ lead to distinct
vertices. Similarly, if the girth is $2r$, then non-returning walks of
length at most $r$ starting with (some directed version of) an edge
$e$ lead to distinct vertices. We will need this observation
again later, so we record it formally.
\begin{observation}\label{ob:moore}
Let $G$ be an undirected simple graph with $n$ vertices and girth $g$.
\begin{description}
\item[\mbox{\em Odd girth:}] Let $g=2r+1$. Then, for all vertices $v$,
$n \geq n_0(v) + n_1(v) + \cdots + n_r(v)$.
\item[\mbox{\em Even girth:}] Let $g=2r$. Let $e$ be an edge of $G$ and suppose $\de_1$ and $\de_2$ are its 
directed versions in $\dG$. Then, \[ n \geq \sum_{i=0}^{r-1} [n_i(\de_1) + n_i(\de_2)].\]
\end{description}
\end{observation}

\begin{proof}[Proof of Theorem \ref{thm:moore}] The claim follows immediately from Observation~\ref{ob:moore} by noting that 
for such a graph $G$, for all vertices $v \in V$ and edges
$\de \in \vec{E}$,
\begin{eqnarray}
\label{eq:moore1}n_i(v) & \geq & \delta(\delta-1)^{i-1} \ \mbox{(for $i\geq 1$)},\ \  n_0(v)=1;\\
\label{eq:moore2}n_i(\de) & \geq & (\delta-1)^i \ \ \mbox{(for $i\geq 0$)}.
\end{eqnarray}
\end{proof}

\subsection{The Alon-Hoory-Linial bound}
Alon, Hoory and Linial showed that the bound in
Theorem~\ref{thm:moore} holds for any undirected graph even when the
minimum degree $\delta$ is replaced by the average degree $\avgd$.
\label{subsec:ahl}
\begin{theorem}[Alon, Hoory and Linial~\cite{AHL02}] \label{thm:ahl}
Let $G$ be an undirected simple graph with $n$ vertices, girth $g$, minimum degree at least $2$ and
average degree $\avgd$.
\begin{description}
\item[\mbox{\em Odd girth:}] If $g=2r+1$, then $\displaystyle n\geq 1+\avgd\sum_{i=0}^{r-1}(\avgd-1)^{i}$.
\item[\mbox{\em Even girth:}] If $g=2r$, then $\displaystyle n\geq 2\sum_{i=0}^{r-1} (\avgd-1)^i$.
\end{description}
\end{theorem}
We will first prove this theorem assuming the following lemma, which is the main technical part of Alon, Hoory and 
Linial~\cite{AHL02}. This lemma shows that the bounds (\ref{eq:moore1}) and (\ref{eq:moore2}) holds with $\delta$ replaced by $\avgd$. 
In Section \ref{subsec:entropy}, we will present an information theoretic proof of this lemma.
\begin{lemma} \label{lm:ahl}
Let $G$ be an undirected simple graph with $n$ vertices,
girth $g$, minimum degree at least two and average degree $\avgd$.
\begin{description} 
\item[(a)] If $v \in V(G)$ is chosen with distribution $\pi$, where
$\pi(v)=d_v/(2|E(G)|) = d_v/(\avgd n)$,
then $\E[n_i(v)] \geq \avgd(\avgd-1)^{i-1}$ $(i \geq 1)$.
\item[(b)] If $\de$ is a uniformly chosen random edge in $\vec{E}$, then
$\E[n_i(\de)] \geq  (\avgd-1)^i$ $(i \geq 0)$.
\end{description}
\end{lemma}
\begin{proof}[Proof of Theorem~\ref{thm:ahl}]
First, consider graphs with odd girth. From Observation~\ref{ob:moore},
Lemma~\ref{lm:ahl} (a) and linearity of expectation we obtain
\[ n \geq \E[n_0(v)+n_1(v) + \cdots + n_r(v)]  \geq 1+\avgd\sum_{i=0}^{r-1}(\avgd-1)^{i},\]
where $v \in V(G)$ is chosen with distribution $\pi$ (defined
in Lemma~\ref{lm:ahl} (a)). 

Now, consider graphs with even girth. Let $\de_1$ be chosen uniformly at random from $\vec{E}$ and let $\de_2$ be its
companion edge (going in the opposite direction). Note that $\de_2$ is also uniformly distributed in $\vec{E}$. Then, from Observation~\ref{ob:moore},
Lemma~\ref{lm:ahl} (b) and linearity of expectation we obtain
\[ n \geq \E\left[ 
\sum_{i=0}^r [n_i(\de_1) + n_i(\de_2)]
\right]  ~\geq~ 2\sum_{i=0}^{r-1} (\avgd-1)^i.
\]
\end{proof}

\subsection{The entropy based proof of Lemma~\ref{lm:ahl}}
\label{subsec:entropy}
The proof of Lemma~\ref{lm:ahl} below is essentially the same as the
one originally proposed by Alon, Hoory and Linial, but is stated in
the language of entropy where some of the arguments based on concavity
are explained directly in information theoretic terms.

\begin{proof}[Proof of Lemma~\ref{lm:ahl}] 
\begin{description}
\item[(a)] Consider the Markov process $v$, $\de_1$ ,$\de_2$, \ldots,
$\de_i$, where $v$ is a random vertex of $G$ chosen with distribution
$\pi$, $\de_1$ is a random edge of $\dG$ leaving $v$ (chosen uniformly
from the $d_v$ choices), and for $1\leq j < i$, $\de_{j+1}$ is a random
successor edge for $\de_j$ chosen uniformly from among the
non-returning possibilities. (If $\de_j$ has the form $(x,y)$, then
there are $d_y-1$ possibilities for $\de_{j+1}$.) Let $v_0=v, v_1,
v_2, \ldots, v_i$ be the vertices visited by this non-returning
walk.  We observe that each $\de_j$ is distributed uniformly in the
set $E(\dG)$ and each $v_j$ has distribution $\pi$.  Then,
\begin{eqnarray*}
\log \E[n_i(v)] & \geq & \E[\log n_i(v)] \\
               & \geq & H[\de_1\de_2 \ldots \de_i \mid v] \\
               &  =   & H[\de_1 | v ] + \sum_{j=1}^{i-1} H[\de_{j+1}  \mid \de_1\de_2 \cdots \de_{j}v]\\
               &  =   & \E[\log d_v]  + \sum_{j=1}^{i-1} \E[\log (d_{v_j}-1)]\\
               &  =   & \E[\log d_v(d_v -1)^{i-1}]\\
               &  =   & \frac{1}{\avgd n} \sum_{v} d_v\log d_v(d_v -1)^{i-1}\\  
               &  \geq   & \log \avgd (\avgd -1)^{i-1},
\end{eqnarray*}
where to justify the first inequality we use Jensen's inequality for
the concave function $\log$, to justify the second we use the fact
that the entropy of a random variable is at most the log of the size
of its support, and to justify the last we use Jensen's inequality for
the convex function $x \log x (x-1)^{i-1}$ ($x\geq 2$). The claim
follows by exponentiating both sides.

\item[(b)] This time we consider the Markov process $\de_0=\de$,
$\de_1$, \ldots, $\de_i$, where $\de$ is chosen uniformly at random
from $\vec{E}$, and for $0 \leq j < i$, $\de_{j+1}$ is a random successor
edge for $\de_j$ chosen uniformly from among the non-returning
possibilities.  Let $v_0, v_1, v_2, \ldots, v_{i+1}$ be the vertices visited 
by this non-returning walk.  As before observe that
each $v_j$ has distribution $\pi$. Then,
\begin{eqnarray*}
\log \E[n_i(e)] & \geq & \E[\log n_i(e)] \\
                & \geq & H[\de_1\de_2 \ldots \de_{i} \mid \de_0] \\
                &  =   & \sum_{j=1}^{i} \E[\log (d_{v_j}-1)]\\
                &  =   & \E[\log (d_{v_0} -1)^i]\\
                &  = & \frac{1}{\avgd n} \sum_v d_v\log (d_v -1)^{i}\\
                &  \geq   & \log (\avgd -1)^{i},
\end{eqnarray*}
where we justify the first two inequalities as before, and the last using
Jensen's inequality applied to the convex function
$x \log (x-1)^{i}$ ($x\geq 2$). The claim follows by exponentiating
both sides.
\end{description}
\end{proof}
\begin{remark}
 Theorem~\ref{thm:ahl} holds for any graph with average degree is at least 2.
For details, see the proof of Theorem 1 in~\cite{AHL02}.
\end{remark} 

\section{Moore bound for bipartite graphs}
\label{sec:moore_bipartite}
Following the proof technique of~\cite{AHL02}, Hoory~\cite{Hoory02} obtained an improved Moore bound for bipartite graphs.
In this section, we provide an information theoretic proof of the same.
\subsection{The Hoory bound}
\begin{theorem}[Hoory~\cite{Hoory02}] \label{thm:h}
Let $G=(V_L,V_R,E)$ be a bipartite graph of girth $g=2r$, with $n_L=|V_L|$ and $n_R=|V_R|$,
minimum degree at least $2$ and the left and right average degrees $d_L$ and $d_R$. Then,
\begin{eqnarray*}
\displaystyle n_L&\geq& \sum_{i=0}^{r-1}(d_R-1)^{\lceil \frac{i}{2}\rceil}(d_L-1)^{\lfloor \frac{i}{2}\rfloor},\\
\displaystyle n_R&\geq& \sum_{i=0}^{r-1}(d_L-1)^{\lceil \frac{i}{2}\rceil}(d_R-1)^{\lfloor \frac{i}{2}\rfloor}.
\end{eqnarray*}
\end{theorem}

For bipartite graphs the girth is always even. We have the following then have the following
variant of Observation~\ref{ob:moore}.
\begin{observation}\label{ob:moore_bipartite}
Let $G=(V_L,V_R,E)$ be an undirected bipartite graph with $|V_L|=n_L$ and $|V_R|=n_R$ and girth $g=2r$.
Let $e$ be an edge of $G$ and suppose $\de_1$ and $\de_2$ be its directed
versions in $\dG$, such that $\de_1 \in \Elr$ and $\de_2 \in \Erl$. Then,
\[ n_L \geq \sum_{i=0}^{\lfloor \frac{r-2}{2} \rfloor} n_{2i+1}(\de_1) + 
\sum_{i=0}^{\lceil \frac{r-2}{2} \rceil} n_{2i}(\de_2). \]
\end{observation}

We will prove the Theorem \ref{thm:h}, assuming the following lemma, which is the main technical part of 
Hoory~\cite{Hoory02}. In Section \ref{subsec:entropy_h}, we will present the proof of this lemma using the language of entropy.

\begin{lemma} \label{lm:h}
Let $G=(V_L,V_R,E)$ be an undirected simple bipartite graph with $n_L$ vertices on the left and $n_R$ 
vertices on the right,
girth $g$, minimum degree at least two and average left and right degrees respectively $d_L$ and $d_R$.
\begin{description} 
\item[(a)] If $\de$ is a uniformly chosen random edge in $\Elr$, 
then $\E[n_{2i+1}(\de)] \geq  (d_R-1)^{i+1}(d_L-1)^{i}$ $(i \geq 1)$.
\item[(b)] If $\de$ is a uniformly chosen random edge in $\Erl$, 
then $\E[n_{2i}(\de)] \geq  (d_R-1)^{i}(d_L-1)^{i}$ $(i \geq 1)$.
\end{description}
\end{lemma}
\begin{proof}[Proof of Theorem~\ref{thm:h}]
We will prove the bound for $n_L$. The proof for $n_R$ case is similar. Let $\de_1$ be chosen uniformly at random from 
$\Elr$ and let $\de_2$ be its companion edge (going in the opposite direction). Note that $\de_2$ is also uniformly
distributed in $\Erl$. Then, from Observation~\ref{lm:h},
Lemma~\ref{lm:h} and linearity of expectation we obtain
\[n_L \geq  \E\left[ 
\sum_{i=0}^{\lfloor \frac{r-2}{2} \rfloor} n_{2i+1}(\de_1) + 
\sum_{i=0}^{\lceil \frac{r-2}{2} \rceil} n_{2i}(\de_2)
\right]  ~\geq~ \sum_{i=0}^{r-1}(d_R-1)^{\lceil \frac{i}{2}\rceil}(d_L-1)^{\lfloor \frac{i}{2}\rfloor}.
\]

\end{proof}
\subsection{The entropy based proof of Lemma \ref{lm:h}}
\label{subsec:entropy_h}
The proof of Lemma~\ref{lm:h} below is essentially the same as the
one originally proposed by Hoory, but is stated in
the language of entropy where some of the arguments based on concavity
are explained directly in information theoretic terms.

\begin{proof}[Proof of Lemma~\ref{lm:h}] 
\begin{description} 
\item[(a)] Consider a Markov process $\de_0, \de_1, \de_2, \cdots,
\de_{2i+1}$, where $\de_0$ is a uniformly chosen random edge from
$\Elr$, and for $0 \leq j < 2i+1$, $\de_{j+1}$ is a random successor
edge for $\de_j$ chosen uniformly from among the non-returning
possibilities. Let $v_0, v_1, v_2, \ldots, v_{2i+2}$ be the vertices
visited by this non-returning walk. We observe that for $0\leq j \leq
i$ each $\de_{2j}$ and $\de_{2j+1}$ is respectively distributed
uniformly in the set $\Elr$ and $\Erl$.  Furthermore, for $j$ even,
$\Pr[v_j =v] = d_{v}/|E(G)|$ for all $v \in V_L$, and for $j$ odd,
$\Pr[v_j =v] = d_{v}/|E(G)|$ for all $v \in V_R$.  Then,
\begin{eqnarray*}
\log \E[n_{2i+1}(e)] & \geq & \E[\log n_{2i+1}(e)] \\
                & \geq & H[\de_0 \de_1 \ldots \de_{2i+1} \mid \de_0] \\
                &  =   & \sum_{j=0}^{i}H[\de_{2j+1}|\de_{2j}]
                        +\sum_{j=1}^{i}H[\de_{2j}|\de_{2j-1}]\\
                &  =   & \sum_{j=0}^{i} \E[\log (d_{v_{2j+1}}-1)]
                        + \sum_{j=1}^{i} \E[\log (d_{v_{2j}}-1)]\\
                &  \geq   & (i+1)\log (d_R-1) + i\log(d_L-1)\\
                &  =   & \log (d_R-1)^{i+1}(d_L-1)^{i}.
\end{eqnarray*}
where to justify the first inequality we use Jensen's inequality for
the concave function $\log$, to justify the second we use the fact
that the entropy of a random variable is at most the log of the size
of its support, and to justify the last we use Jensen's inequality for
the convex function $x \log (x-1)$ ($x\geq 2$). The claim
follows by exponentiating both sides.

\item[(b)] Similarly,
  \begin{equation*}
    \log \E[n_{2i}(e)] \geq \log (d_L-1)^i(d_R-1)^i.
  \end{equation*}
\comment{Consider a Markov process $\de_0, \de_1, \cdots, \de_{2i}$,
where $\de_0$ is a uniformly chosen random edge from $\Erl{RL}$, and
for $0 \leq j < 2i$, $\de_{j+1}$ is a random successor edge for
$\de_j$ chosen uniformly from among the non-returning
possibilities. Let $v_0, v_1, v_2, \ldots, v_{2i+1}$ be the vertices
visited by this non-returning walk.  We observe that for $0\leq j \leq
i$ each $\de_{2j}$ and $\de_{2j+1}$ is respectively distributed
uniformly in the set $\Erl{RL}(\dG)$ and $\Elr(\dG)$.  Furthermore, for
$j$ even, $\Pr[v_j =v] = d_{v}/|E(G)|$ for all $v \in V_L$, and
for $j$ odd, $\Pr[v_j =v] = d_{v}/|E(G)|$ for all $v \in V_R$.
  Then,
\begin{eqnarray*}
\log \E[n_{2i}(e)] & \geq & \E[\log n_{2i}(e)] \\
                & \geq & H[\de_0\de_1 \ldots \de_{2i} \mid \de_0] \\
                &  =   & \sum_{j=0}^{i-1}H[\de_{2j+1}|\de_{2j}]
                        +\sum_{j=1}^{i}H[\de_{2j}|\de_{2j-1}]\\
                &  =   & \sum_{j=0}^{i-1} \E[\log (d_{v_{2j+1}}-1)]
                        + \sum_{j=1}^{i} \E[\log (d_{v_{2j}}-1)]\\
                &  \geq   & i\log (d_L-1) + i\log(d_R-1)\\
                &  =   & \log (d_L-1)^i(d_R-1)^i.
\end{eqnarray*}}
\end{description}
\end{proof}

\comment{
\subsection{Moore Bound for Hypergraphs}
A hypergraph $\mathcal{H}=(V,E)$ can be represented by an undirected bipartite graphs $G=(V,E,E_G)$, where
$\{v,e\} \in E_G$ iff $v \in e$ in $\mathcal{H}$. Thus if $\mathcal{H}$ has girth $g$, average cardinality 
 $\bar{c}$ and average degree $\avgd$, then from theorem \ref{thm:h} we get 
\[|V| \geq \sum_{i=0}^{g-1}(\bar{c}-1)^{\lceil \frac{i}{2}\rceil}(\avgd-1)^{\lfloor \frac{i}{2}\rfloor}, \]}

\bibliographystyle{plain}
\bibliography{bibfile}

\begin{thebibliography}{1}

\bibitem{AHL02}
Noga Alon, Shlomo Hoory, and Nathan Linial.
\newblock The {M}oore bound for irregular graphs.
\newblock {\em Graphs and Combinatorics}, 18(1):53--57, 2002.

\bibitem{B93}
N.~Biggs.
\newblock {\em Algebraic graph theory}.
\newblock Cambridge University Press, Cambridge, second edition, 1993.

\bibitem{Hoory02}
Shlomo Hoory.
\newblock The size of bipartite graphs with a given girth.
\newblock {\em J. Comb. Theory, Ser. B}, 86(2):215--220, 2002.

\end{thebibliography}
\end{document}